\documentclass[journal,draftcls,onecolumn,11pt]{IEEEtran}
\usepackage{epsfig, amssymb}
\usepackage[cmex10]{amsmath}
\usepackage{subfigure}
\usepackage{subfigure,cite,graphicx,amsmath,amssymb,mathrsfs,epsf, color}

\newtheorem{theorem}{Theorem}%[section]
\newtheorem{remark}{Remark}

\newtheorem{corollary}{Corollary}%[section]
\DeclareMathOperator\C{\sf C}

\def \P{\operatorname{Pr}}

\begin{document}
\title{A New Achievable Scheme for Interference Relay Channels}
\author{\authorblockN{Byungjun Kang, Si-Hyeon Lee, Sae-Young Chung and Changho Suh\\}
\authorblockA{Department of EE, KAIST, Daejeon, Korea\\
Email: bj\_kang@kaist.ac.kr, sihyeon@kaist.ac.kr,
sychung@ee.kaist.ac.kr, chsuh@ee.kaist.ac.kr} } \maketitle

\begin{abstract}
%We establish an achievable rate region for discrete memoryless interference multi-relay channels that consist of two source-destination pairs and an arbitrary number of relays. We develop an %achievable scheme combining Han-Kobayashi and noisy network coding schemes. We apply our schemes in two cases. First, we characterize the capacity region of a class of
%discrete memoryless interference relay channels. This class naturally generalizes the injective deterministic discrete memoryless interference channel by El Gamal and Costa and the %deterministic discrete memoryless relay channel with orthogonal receiver components by Kim. Moreover, for the Gaussian interference relay channel with orthogonal receiver components, we %characterize its capacity region to within a number of bits independent of signal-to-noise ratio and interference-to-noise ratio. For this channel, we also show that our gap is smaller than %that of Zhou and Yu.
We establish an achievable rate region for discrete memoryless interference relay channels that consist of two source-destination pairs and one or more relays. We develop an achievable scheme combining Han-Kobayashi and noisy network coding schemes. We apply our achievability to two cases. First, we characterize the capacity region of a class of
discrete memoryless interference relay channels. This class naturally generalizes the injective deterministic discrete memoryless interference channel by El Gamal and Costa and the deterministic discrete memoryless relay channel with orthogonal receiver components by Kim. Moreover, for the Gaussian interference relay channel with orthogonal receiver components, we show that our scheme achieves a better sum rate than that of noisy network coding.
\end{abstract}

\IEEEpeerreviewmaketitle
\allowdisplaybreaks
%%%%%%%%%%%%%%%%%%%%%%%%%%%%%%%%%%%%%%%%%%%%%%
%%%%%%%%%% Introduction
%%%%%%%%%%%%%%%%%%%%%%%%%%%%%%%%%%%%%%%%%%%%%%
\section{Introduction}
Discrete memoryless interference channel (DM-IC) was introduced by Ahlswede \cite{Ahlswede:74}. Discrete memoryless relay channel (DM-RC) was first studied by van der Meulen \cite{Meulen:71}. Neither the capacity region of the DM-IC nor the capacity of the DM-RC has been characterized yet except for some special cases. First, for DM-IC, the best known inner bound was obtained by Han and Kobayashi \cite{HanKobayashi:81}. This inner bound was shown to be tight for the injective deterministic DM-IC by El Gamal and Costa \cite{GamalCosta:82}. On the other hand, one relaying strategy for DM-RC is compress-and-forward (CF) due to Cover and El Gamal \cite{Gamal:79} where the relay compresses its observation and forwards it to the destination. CF was shown to be optimal for the deterministic DM-RC with orthogonal receiver components \cite{Kim:08} and the modulo-2 sum relay channel \cite{Yu:09}. Recently, noisy network coding \cite{Lim:10} generalized CF for general discrete memoryless relay networks.

A natural next step is to extend these results to more general channel scenarios in which there are more than two transmitter-receiver pairs and/or relays. As one model, we consider a discrete memoryless interference multi-relay channel (DM-IMRC) that consists of two source-destination pairs and an arbitrary number of relays. For this channel, we combine Han-Kobayashi and noisy network coding schemes to establish an achievable rate region. We apply our result to two cases. First, we characterize the capacity region of a class of a discrete memoryless interference relay channel (DM-IRC) which naturally generalizes the injective deterministic DM-IC by El Gamal and Costa \cite{GamalCosta:82} and the deterministic DM-RC with orthogonal receiver components by Kim \cite{Kim:08}. For the converse, a genie-aided proof technique is used. Furthermore, for a Gaussian interference relay channel (GIRC) with orthogonal receiver components, we show we can obtain a better sum rate than that in \cite{Lim:10}.

%A natural next step is to extend these results to more general channel scenarios in which there are more than two transmitter-receiver pairs and/or relays. As one model, we consider a discrete memoryless interference multi-relay channel (DM-IMRC) that consists of two source-destination pairs and an arbitrary number of relays. For this channel, we combine Han-Kobayashi and noisy network coding schemes to establish an achievable rate region. We apply our result to two cases. First, we characterize the capacity region of a class of DM-IRC which naturally generalizes the injective deterministic DM-IC by El Gamal and Costa \cite{GamalCosta:82} and the deterministic DM-RC with orthogonal receiver components by Kim \cite{Kim:08}. Furthermore, for a Gaussian interference relay channel (IRC) with orthogonal receiver components, we obtain better sum rates than that in \cite{Lim:10}. Moreover, we characterize the capacity region of the Gaussian IRC to within a number of bits independent of SNR and INR. We compare our result with that of Zhou and Yu \cite{Yu:11} which uses Han-Kobayashi and Hash-and-Forward schemes, and show that our gap is smaller.

The rest of the paper is organized as follows. In Section \ref{sec:model}, we introduce the DM-IMRC model. Section \ref{sec:result} presents the achievable rate region for the DM-IMRC. Section \ref{sec:injec} characterizes the capacity region of a class of DM-IRC. Section \ref{sec:gap} focuses on the GIRC with orthogonal receiver components.\\

\section{Model} \label{sec:model}
We consider a DM-IMRC as depicted in Fig. \ref{fig:DMIRC}.
A $\left(2^{nR_1}, 2^{nR_2}, n\right)$ code consists of two message sets $\mathcal{M}_1=\{1,\ldots,2^{nR_1}\}$ and  $\mathcal{M}_2=\{1,\ldots,2^{nR_2}\}$, two encoding functions at the sources where the first source (node 1) maps its message $m_1\in\mathcal{M}_1$ to a codeword $x_1^n(m_1)\in \mathcal{X}_1^n$ and the second source (node 2) maps its message $m_2\in\mathcal{M}_2$ to a codeword $x_2^n(m_2)\in \mathcal{X}_2^n$, $N$ processing functions at the relays (node 3,$k$ where $k\in [1:N]$) that map each past received symbols $y_{3,k}^{i-1}\in \mathcal{Y}_{3,k}^{i-1}$ to a symbol $x_{3,k,i}(y_{3,k}^{i-1})\in \mathcal{X}_{3,k}$, and two decoding functions at the destinations where the first destination (node 4) maps each received sequence $y_4^n\in\mathcal{Y}_4^n$ to a message estimate $\hat{m}_1$ and the second destination (node 5) maps each received sequence $y_5^n\in\mathcal{Y}_5^n$ to a message estimate $\hat{m}_2$. The first source (node 1) chooses an index $m_1$ uniformly from the set $\mathcal{M}_1$ and sends $x_1^n(m_1)$ and the second source (node 2) chooses an index $m_2$ uniformly from the set $\mathcal{M}_2$ and sends $x_2^n(m_2)$. The average probability of error for a $(2^{nR_1}, 2^{nR_2}, n)$ code is given as $P_e^{(n)}\triangleq  \P\left((\hat{M}_1,\hat{M}_2) \neq ( M_1,M_2) \right)$. A rate pair $(R_1, R_2)$ is said to be \emph{achievable} if there exists a sequence of $(2^{nR_1}, 2^{nR_2}, n)$ codes such that $P_e^{(n)}\rightarrow 0$ as $n\rightarrow \infty$. The capacity region $C$ is the closure of the set of achievable rate pairs $(R_1, R_2)$.\\
\begin{figure}[t]
 \centering
  { \includegraphics[width=100mm]{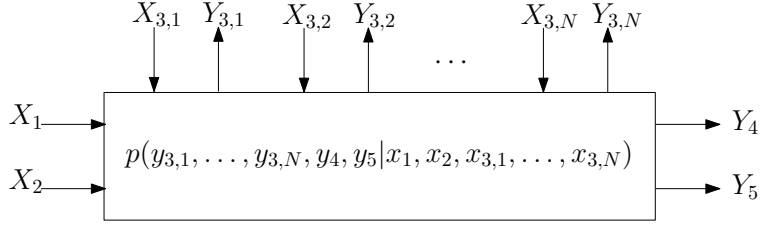}}
  \caption{A discrete memoryless interference multi-relay channel (DM-IMRC)} \label{fig:DMIRC}
\end{figure}

\section{Main results} \label{sec:result}
An achievable rate region for the DM-IMRC is established in the following theorem.
\begin{theorem}\label{thm:mib}
A rate pair $(R_1, R_2)$ is achievable for the DM-IMRC if there exists some probability mass function (pmf) $p(q)p(u_1,x_1|q)p(u_2,x_2|q)\prod_{k=1}^N p(x_{3,k}|q)p(\hat{y}_{3,k}|y_{3,k},x_{3,k},q)$ such that
\begin{align*}
R_1<&\min_{S}\{I(X_1,X_3(S);\hat{Y}_3(S^c),Y_4|U_2,X_3(S^c))-I(\hat{Y}_3(S);Y_3(S)|X_1,U_2,X_3^N,\hat{Y}_3(S^c),Y_4)\}\\
R_2<&\min_{S}\{I(X_2,X_3(S);\hat{Y}_3(S^c),Y_5|U_1,X_3(S^c))-I(\hat{Y}_3(S);Y_3(S)|X_2,U_1,X_3^N,\hat{Y}_3(S^c),Y_5)\}\\
R_1+R_2<&\min_{S}\{I(X_1,X_3(S);\hat{Y}_3(S^c),Y_4|U_1,U_2,X_3(S^c))-I(\hat{Y}_3(S);Y_3(S)|X_1,U_2,X_3^N,\hat{Y}_3(S^c),Y_4)\}\\
+&\min_{S}\{I(X_2,U_1,X_3(S);\hat{Y}_3(S^c),Y_5|X_3(S^c))-I(\hat{Y}_3(S);Y_3(S)|X_2,U_1,X_3^N,\hat{Y}_3(S^c),Y_5)\}\\
R_1+R_2<&\min_{S}\{I(X_2,X_3(S);\hat{Y}_3(S^c),Y_5|U_1,U_2,X_3(S^c))-I(\hat{Y}_3(S);Y_3(S)|X_2,U_1,X_3^N,\hat{Y}_3(S^c),Y_5)\}\\
+&\min_{S}\{I(X_1,U_2,X_3(S);\hat{Y}_3(S^c),Y_4|X_3(S^c))-I(\hat{Y}_3(S);Y_3(S)|X_1,U_2,X_3^N,\hat{Y}_3(S^c),Y_4)\}\\
R_1+R_2<&\min_{S}\{I(X_1,U_2,X_3(S);\hat{Y}_3(S^c),Y_4|U_1,X_3(S^c))-I(\hat{Y}_3(S);Y_3(S)|X_1,U_2,X_3^N,\hat{Y}_3(S^c),Y_4)\}\\
+&\min_{S}\{I(X_2,U_1,X_3(S);\hat{Y}_3(S^c),Y_5|U_2,X_3(S^c))-I(\hat{Y}_3(S);Y_3(S)|X_2,U_1,X_3^N,\hat{Y}_3(S^c),Y_5)\}\\
2R_1+R_2<&\min_{S}\{I(X_1,X_3(S);\hat{Y}_3(S^c),Y_4|U_1,U_2,X_3(S^c))-I(\hat{Y}_3(S);Y_3(S)|X_1,U_2,X_3^N,\hat{Y}_3(S^c),Y_4)\}\\
+&\min_{S}\{I(X_1,U_2,X_3(S);\hat{Y}_3(S^c),Y_4|X_3(S^c))-I(\hat{Y}_3(S);Y_3(S)|X_1,U_2,X_3^N,\hat{Y}_3(S^c),Y_4)\}\\
+&\min_{S}\{I(X_2,U_1,X_3(S);\hat{Y}_3(S^c),Y_5|U_2,X_3(S^c))-I(\hat{Y}_3(S);Y_3(S)|X_2,U_1,X_3^N,\hat{Y}_3(S^c),Y_5)\}\\
R_1+2R_2<&\min_{S}\{I(X_2,X_3(S);\hat{Y}_3(S^c),Y_5|U_1,U_2,X_3(S^c))-I(\hat{Y}_3(S);Y_3(S)|X_2,U_1,X_3^N,\hat{Y}_3(S^c),Y_5)\}\\
+&\min_{S}\{I(X_2,U_1,X_3(S);\hat{Y}_3(S^c),Y_5|X_3(S^c))-I(\hat{Y}_3(S);Y_3(S)|X_2,U_1,X_3^N,\hat{Y}_3(S^c),Y_5)\}\\
+&\min_{S}\{I(X_1,U_2,X_3(S);\hat{Y}_3(S^c),Y_4|U_1,X_3(S^c))-I(\hat{Y}_3(S);Y_3(S)|X_1,U_2,X_3^N,\hat{Y}_3(S^c),Y_4)\}
\end{align*}
for all subsets $S\subset[1:N]$ such that $X_3(S)\subset\{X_{3,1},\cdots,X_{3,N}\}$ which are relay nodes.
\end{theorem}
\begin{proof}
See Appendix \ref{app:mib}.\\
\end{proof}

By letting $N=1$ in Theorem \ref{thm:mib}, we obtain the following achievable rate region for the DM-IRC.
\begin{corollary}\label{cor:ib}
A rate pair $(R_1, R_2)$ is achievable for the DM-IRC if
\begin{align*}
R_1<&\min\{I(X_1;\hat{Y}_3,Y_4|U_2,X_3,Q)), I(X_1,X_3;\!Y_4|U_2,Q)\!-\!I(\hat{Y}_3;\!Y_3|X_1,U_2,X_3,Y_4,Q)\}\\
R_2<&\min\{I(X_2;\hat{Y}_3,Y_5|U_1,X_3,Q), I(X_2,X_3;\!Y_5|U_1,Q)\!-\!I(\hat{Y}_3;\!Y_3|X_2,U_1,X_3,Y_5,Q)\}\\
R_1\!+\!R_2<&\min\{I(X_1;\hat{Y}_3,Y_4|U_1,U_2,X_3,Q), I(X_1,X_3;Y_4|U_1,U_2,Q)-I(\hat{Y}_3;Y_3|X_1,U_2,X_3,Y_4,Q)\}\\
+&\min\{I(X_2,U_1;\hat{Y}_3,Y_5|X_3,Q), I(X_2,U_1,X_3;Y_5|Q)-I(\hat{Y}_3;Y_3|X_2,U_1,X_3,Y_5,Q)\}\\
R_1\!+\!R_2<&\min\{I(X_2;\hat{Y}_3,Y_5|U_1,U_2,X_3,Q), I(X_2,X_3;Y_5|U_1,U_2,Q)-I(\hat{Y}_3;Y_3|X_2,U_1,X_3,Y_5,Q)\}\\
+&\min\{I(X_1,U_2;\hat{Y}_3,Y_4|X_3,Q), I(X_1,U_2,X_3;Y_4|Q)-I(\hat{Y}_3;Y_3|X_1,U_2,X_3,Y_4,Q)\}\\
R_1\!+\!R_2<&\min\{I(X_1,U_2;\hat{Y}_3,Y_4|U_1,X_3,Q), I(X_1,\!U_2,\!X_3;\!Y_4|U_1,Q)-I(\hat{Y}_3;Y_3|X_1,U_2,X_3,Y_4,Q)\}\\
+&\min\{I(X_2,U_1;\hat{Y}_3,Y_5|U_2,X_3,Q), I(X_2,\!U_1,\!X_3;\!Y_5|U_2,Q)-I(\hat{Y}_3;Y_3|X_2,U_1,X_3,Y_5,Q)\}\\
2R_1\!+\!R_2<&\min\{I(X_1;\hat{Y}_3,Y_4|U_1,U_2,X_3,Q), I(X_1,X_3;Y_4|U_1,U_2,Q)-I(\hat{Y}_3;Y_3|X_1,U_2,X_3,Y_4,Q)\}\\
+&\min\{I(X_1,U_2;\hat{Y}_3,Y_4|X_3,Q), I(X_1,U_2,X_3;Y_4|Q)-I(\hat{Y}_3;Y_3|X_1,U_2,X_3,Y_4,Q)\}\\
+&\min\{I(X_2,U_1;\hat{Y}_3,Y_5|U_2,X_3,Q), I(X_2,\!U_1,\!X_3;\!Y_5|U_2,Q)-I(\hat{Y}_3;Y_3|X_2,U_1,X_3,Y_5,Q)\}\\
R_1\!+\!2R_2<&\min\{I(X_2;\hat{Y}_3,Y_5|U_1,U_2,X_3,Q), I(X_2,X_3;Y_5|U_1,U_2,Q)-I(\hat{Y}_3;Y_3|X_2,U_1,X_3,Y_5,Q)\}\\
+&\min\{I(X_2,U_1;\hat{Y}_3,Y_5|X_3,Q), I(X_2,U_1,X_3;Y_5|Q)-I(\hat{Y}_3;Y_3|X_2,U_1,X_3,Y_5,Q)\}\\
+&\min\{I(X_1,U_2;\hat{Y}_3,Y_4|U_1,X_3,Q), I(X_1,\!U_2,\!X_3;\!Y_4|U_1,Q)-I(\hat{Y}_3;Y_3|X_1,U_2,X_3,Y_4,Q)\}
\end{align*}
for some pmf $p(q)p(u_1,x_1|q)p(u_2,x_2|q)p(x_3|q)p(\hat{y}_3|y_3,x_3,q)$.\\
\end{corollary}

\section{Capacity of a class of injective interference relay channels} \label{sec:injec}
We characterize the capacity region of a class of injective DM-IRCs. In this class, the channel outputs are given as follows:
\begin{align*}
Y_4&=(Y_4',Y_4'')\\
Y_5&=(Y_5',Y_5'')\\
Y_4'&=y_4(X_1,T_2)\\
Y_5'&=y_5(X_2,T_1)\\
Y_4''&=Y_5''=X_3\\
Y_3&=f_1(X_1,Y_4')=f_2(X_2,Y_5')
\end{align*}
where $T_1=t_1(X_1)$ and $T_2=t_2(X_2)$ are functions of $X_1$ and $X_2$, respectively and $f_1$ and $f_2$ are functions of $(X_1,Y_4')$ and $(X_2,Y_5')$, respectively. The functions $y_4$ and $y_5$ are injective in $t_1$ and $t_2$, respectively, i.e., for every $x_1\in \mathcal{X}_1$, $y_4(x_1,t_2)$ is a one-to-one function of $t_2$ and similarly for $y_5$. The relay sends information over a common rate-limited noiseless link of rate $R_0\triangleq\max_{p(x_3)} I(X_3;Y_4^{\prime\prime})=\max_{p(x_3)} I(X_3;Y_5^{\prime\prime})$ to both destinations. This class of DM-IRCs is illustrated in Fig. \ref{fig:DMIRC_cap}.
\begin{figure}[t]
 \centering
  { \includegraphics[width=75mm]{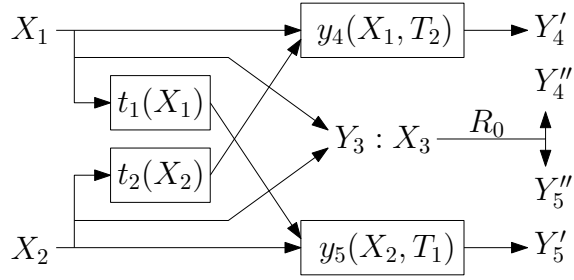}}
  \caption{A class of injective DM-IRCs} \label{fig:DMIRC_cap}
\end{figure}
For the class of injective DM-IRCs illustrated in Fig. \ref{fig:DMIRC_cap}, the following theorem gives the capacity region.
\begin{theorem}\label{thm:capacity}
The capacity region of the class of injective DM-IRCs in Fig. \ref{fig:DMIRC_cap} is the set of rate pairs $(R_1, R_2)$ such that
\begin{align}
R_1\leq&\min\{H(Y_4^{\prime}|T_2,Q)+H(Y_3|Y_4^{\prime},T_2,Q),H(Y_4^{\prime}|T_2,Q)+R_0\}\\
R_2\leq&\min\{H(Y_5^{\prime}|T_1,Q)+H(Y_3|Y_5^{\prime},T_1,Q),H(Y_5^{\prime}|T_1,Q)+R_0\}\\
R_1+R_2\leq&\min\{H(Y_3,Y_4^{\prime}|T_1,T_2,Q),H(Y_4^{\prime}|T_1,T_2,Q)+R_0\}+\min\{H(Y_3,Y_5^{\prime}|Q),H(Y_5^{\prime}|Q)+R_0\}\\
R_1+R_2\leq&\min\{H(Y_3,Y_5^{\prime}|T_1,T_2,Q),H(Y_5^{\prime}|T_1,T_2,Q)+R_0\}+\min\{H(Y_3,Y_4^{\prime}|Q),H(Y_4^{\prime}|Q)+R_0\}\\
R_1+R_2\leq&\min\{H(Y_3,Y_4^{\prime}|T_1,Q),H(Y_4^{\prime}|T_1,Q)+R_0\}+\min\{H(Y_3,Y_5^{\prime}|T_2,Q),H(Y_5^{\prime}|T_2,Q)+R_0\}\\
2R_1+R_2\leq&\min\{H(Y_3,Y_4^{\prime}|T_1,T_2,Q),H(Y_4^{\prime}|T_1,T_2,Q)+R_0\}+\min\{H(Y_3,Y_4^{\prime}|Q),H(Y_4^{\prime}|Q)+R_0\}\nonumber\\
&+\min\{H(Y_3,Y_5^{\prime}|T_2,Q),H(Y_5^{\prime}|T_2,Q)+R_0\}\\
R_1+2R_2\leq&\min\{H(Y_3,Y_5^{\prime}|T_1,T_2,Q),H(Y_5^{\prime}|T_1,T_2,Q)+R_0\}+\min\{H(Y_3,Y_5^{\prime}|Q),H(Y_5^{\prime}|Q)+R_0\}\nonumber\\
&+\min\{H(Y_3,Y_4^{\prime}|T_1,Q),H(Y_4^{\prime}|T_1,Q)+R_0\}
\end{align}
for some pmf $p(q)p(x_1|q)p(x_2|q)$.
\end{theorem}

\begin{proof}
The achievability of Theorem \ref{thm:capacity} is directly obtained by letting $\hat{Y}_3=Y_3, U_1=T_1,$ and $U_2=T_2$ in Corollary \ref{cor:ib}. The converse proof is given in Appendix \ref{app:thm2}.\\
\end{proof}

\section{Gaussian interference relay channel with orthogonal receiver components} \label{sec:gap}
Consider the GIRC with orthogonal receiver components in Fig. \ref{fig:GIRC}. The channel outputs are
\begin{align*}
Y_3&=g_{31}X_1+g_{32}X_2+Z_3\\
Y_4^{\prime}&=g_{41}X_1+g_{42}X_2+Z_4\\
Y_5^{\prime}&=g_{51}X_1+g_{52}X_2+Z_5.
\end{align*}
where $Y_l=(Y_l^{\prime},Y_l^{\prime\prime}), Y_l^{\prime}$ and $Y_l^{\prime\prime}$ are independent for $l=4,5$, $g_{jk}$ is the channel gain from node $k$ to node $j$ and the noise $Z_i\sim \mathcal{N}(0,1)$ is independent and identically distributed (i.i.d.). Relay helps the communication of two source-destination pairs by forwarding some information about $Y_3$ to both destinations through a common rate-limited noiseless link of rate $R_0\triangleq\max_{p(x_3)} I(X_3;Y_4^{\prime\prime})=\max_{p(x_3)} I(X_3;Y_5^{\prime\prime})$.
\begin{figure}[t]
 \centering
  { \includegraphics[width=75mm]{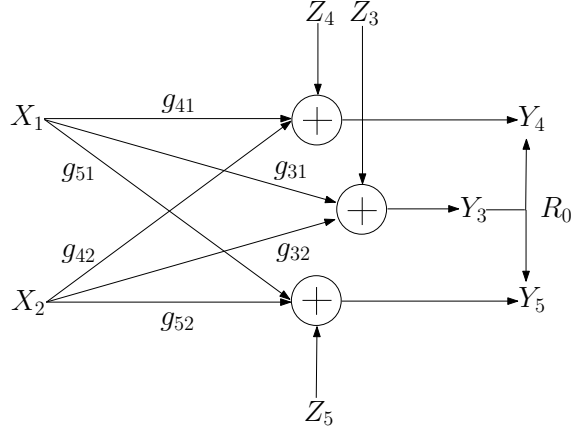}}
  \caption{GIRC with orthogonal receiver components} \label{fig:GIRC}
\end{figure}

We consider $P_1=P_2=P$, $X_1=U_1+V_1$, $X_2=U_2+V_2$ and $U_1$, $V_1$, $U_2$ and $V_2$ are independent where $U_i$ corresponds to the common message and $V_i$ corresponds to the private message for $i=1,2$ and power is allocated as $P_{U_i}=(1-\alpha_i)P, P_{V_i}=\alpha_iP$ for $i=1,2$. Then, setting $\hat{Y}_3=Y_3+\hat{Z}$ with $\hat{Z}\sim \mathcal{N}(0,\sigma^2)$ yields the inner bound $\mathcal{R}$ that consists of all rate pairs $(R_1,R_2)$ such that
\begin{align*}
R_1<&\C\left(\frac{b_{11}P+a_1^2\alpha_2P^2}{1+\sigma^2+b_{12}\alpha_2P}\right)\\
R_1<&\C\left(\frac{g_{41}^2P}{g_{42}^2\alpha_2P+1}\right)-C_1+R_0\\
R_2<&\C\left(\frac{b_{22}P+a_2^2\alpha_1P^2}{1+\sigma^2+b_{21}\alpha_1P}\right)\\
R_2<&\C\left(\frac{g_{52}^2P}{g_{51}^2\alpha_1P+1}\right)-C_2+R_0\\
R_1+R_2<&\C\left(\frac{b_{11}\alpha_1P+a_1^2\alpha_1\alpha_2P^2}{1+\sigma^2+b_{12}\alpha_2P}\right)+\C\left(\frac{b_{21}(1-\alpha_1)P+b_{22}P+a_2^2P^2}{1+\sigma^2+b_{21}\alpha_1P}\right)\\
R_1+R_2<&\C\left(\frac{g_{41}^2\alpha_1P}{g_{42}^2\alpha_2P+1}\right)-C_1+\C\left(\frac{g_{52}^2P+(1-\alpha_1)g_{51}^2P}{g_{51}^2\alpha_1P+1}\right)-C_2+2R_0\\
R_1+R_2<&\C\left(\frac{b_{11}\alpha_1P+a_1^2\alpha_1\alpha_2P^2}{1+\sigma^2+b_{12}\alpha_2P}\right)+\C\left(\frac{g_{52}^2P+(1-\alpha_1)g_{51}^2P}{g_{51}^2\alpha_1P+1}\right)-C_2+R_0\\
R_1+R_2<&\C\left(\frac{b_{21}(1-\alpha_1)P+b_{22}P+a_2^2P^2}{1+\sigma^2+b_{21}\alpha_1P}\right)+\C\left(\frac{g_{41}^2\alpha_1P}{g_{42}^2\alpha_2P+1}\right)-C_1+R_0\\
R_1+R_2<&\C\left(\frac{b_{22}\alpha_2P+a_2^2\alpha_1\alpha_2P^2}{1+\sigma^2+b_{21}\alpha_1P}\right)+\C\left(\frac{b_{12}(1-\alpha_2)P+b_{11}P+a_1^2P^2}{1+\sigma^2+b_{12}\alpha_2P}\right)\\
R_1+R_2<&\C\left(\frac{g_{52}^2\alpha_2P}{g_{51}^2\alpha_1P+1}\right)-C_1+\C\left(\frac{g_{41}^2P+(1-\alpha_2)g_{42}^2P}{g_{42}^2\alpha_2P+1}\right)-C_2+2R_0\\
R_1+R_2<&\C\left(\frac{b_{22}\alpha_2P+a_2^2\alpha_1\alpha_2P^2}{1+\sigma^2+b_{21}\alpha_1P}\right)+\C\left(\frac{g_{41}^2P+(1-\alpha_2)g_{42}^2P}{g_{42}^2\alpha_2P+1}\right)-C_1+R_0\\
R_1+R_2<&\C\left(\frac{b_{12}(1-\alpha_2)P+b_{11}P+a_1^2P^2}{1+\sigma^2+b_{12}\alpha_2P}\right)+\C\left(\frac{g_{52}^2\alpha_2P}{g_{51}^2\alpha_1P+1}\right)-C_2+R_0\\
R_1+R_2<&\C\left(\frac{b_{12}(1-\alpha_2)P+b_{11}\alpha_1P+a_1^2\alpha_1P^2}{1+\sigma^2+b_{12}\alpha_2P}\right)+\C\left(\frac{b_{21}(1-\alpha_1)P+b_{22}\alpha_2P+a_2^2\alpha_2P^2}{1+\sigma^2+b_{21}\alpha_1P}\right)\\
R_1+R_2<&\C\left(\frac{g_{41}^2\alpha_1P+(1-\alpha_2)g_{42}^2P}{g_{42}^2\alpha_2P+1}\right)-C_1+\C\left(\frac{g_{52}^2\alpha_2P+(1-\alpha_1)g_{51}^2P}{g_{51}^2\alpha_1P+1}\right)-C_2+2R_0\\
R_1+R_2<&\C\left(\frac{b_{12}(1-\alpha_2)P+b_{11}\alpha_1P+a_1^2\alpha_1P^2}{1+\sigma^2+b_{12}\alpha_2P}\right)+\C\left(\frac{g_{52}^2\alpha_2P+(1-\alpha_1)g_{51}^2P}{g_{51}^2\alpha_1P+1}\right)-C_2+R_0\\
R_1+R_2<&\C\left(\frac{b_{21}(1-\alpha_1)P+b_{22}\alpha_2P+a_2^2\alpha_2P^2}{1+\sigma^2+b_{21}\alpha_1P}\right)+\C\left(\frac{g_{41}^2\alpha_1P+(1-\alpha_2)g_{42}^2P}{g_{42}^2\alpha_2P+1}\right)-C_1+R_0\\
2R_1+R_2<&\C\left(\frac{b_{11}\alpha_1P+a_1^2\alpha_1\alpha_2P^2}{1+\sigma^2+b_{12}\alpha_2P}\right)+\C\left(\frac{b_{12}(1-\alpha_2)P+b_{11}P+a_1^2P^2}{1+\sigma^2+b_{12}\alpha_2P}\right)\\
&+\C\left(\frac{b_{21}(1-\alpha_1)P+b_{22}\alpha_2P+a_2^2\alpha_2P^2}{1+\sigma^2+b_{21}\alpha_1P}\right)\nonumber\\
2R_1+R_2<&\C\left(\frac{g_{41}^2\alpha_1P}{g_{42}^2\alpha_2P+1}\right)+\C\left(\frac{g_{41}^2P+(1-\alpha_2)g_{42}^2P}{g_{42}^2\alpha_2P+1}\right)+\C\left(\frac{g_{52}^2\alpha_2P+(1-\alpha_1)g_{51}^2P}{g_{51}^2\alpha_1P+1}\right)\\
&+3R_0-2C_1-C_2\nonumber\\
2R_1+R_2<&\C\left(\frac{g_{41}^2\alpha_1P}{g_{42}^2\alpha_2P+1}\right)+\C\left(\frac{g_{41}^2P+(1-\alpha_2)g_{42}^2P}{g_{42}^2\alpha_2P+1}\right)+2R_0-2C_1\\
&+\C\left(\frac{b_{21}(1-\alpha_1)P+b_{22}\alpha_2P+a_2^2\alpha_2P^2}{1+\sigma^2+b_{21}\alpha_1P}\right)\nonumber\\
2R_1+R_2<&\C\left(\frac{b_{11}\alpha_1P+a_1^2\alpha_1\alpha_2P^2}{1+\sigma^2+b_{12}\alpha_2P}\right)+\C\left(\frac{b_{12}(1-\alpha_2)P+b_{11}P+a_1^2P^2}{1+\sigma^2+b_{12}\alpha_2P}\right)\\
&+\C\left(\frac{g_{52}^2\alpha_2P+(1-\alpha_1)g_{51}^2P}{g_{51}^2\alpha_1P+1}\right)+R_0-C_2\nonumber\\
2R_1+R_2<&\C\left(\frac{b_{11}\alpha_1P+a_1^2\alpha_1\alpha_2P^2}{1+\sigma^2+b_{12}\alpha_2P}\right)+\C\left(\frac{g_{41}^2P+(1-\alpha_2)g_{42}^2P}{g_{42}^2\alpha_2P+1}\right)\\
&+\C\left(\frac{g_{52}^2\alpha_2P+(1-\alpha_1)g_{51}^2P}{g_{51}^2\alpha_1P+1}\right)+2R_0-C_1-C_2\nonumber\\
2R_1+R_2<&\C\left(\frac{b_{11}\alpha_1P+a_1^2\alpha_1\alpha_2P^2}{1+\sigma^2+b_{12}\alpha_2P}\right)+\C\left(\frac{b_{21}(1-\alpha_1)P+b_{22}\alpha_2P+a_2^2\alpha_2P^2}{1+\sigma^2+b_{21}\alpha_1P}\right)\\
&+\C\left(\frac{g_{41}^2P+(1-\alpha_2)g_{42}^2P}{g_{42}^2\alpha_2P+1}\right)+R_0-C_1\nonumber
%R_1+2R_2<&\C\left(\frac{b_{22}\alpha_2P+a_2^2\alpha_1\alpha_2P^2}{1+\sigma^2+b_{21}\alpha_1P}\right)+\C\left(\frac{b_{21}(1-\alpha_1)P+b_{22}P+a_2^2P^2}{1+\sigma^2+b_{21}\alpha_1P}\right)\\
%&+\C\left(\frac{b_{12}(1-\alpha_2)P+b_{11}\alpha_1P+a_1^2\alpha_1P^2}{1+\sigma^2+b_{12}\alpha_2P}\right)\nonumber\\
%R_1+2R_2<&\C\left(\frac{g_{52}^2\alpha_2P}{g_{51}^2\alpha_1P+1}\right)+\C\left(\frac{g_{52}^2P+(1-\alpha_1)g_{51}^2P}{g_{51}^2\alpha_1P+1}\right)+\C\left(\frac{g_{41}^2\alpha_1P+(1-\alpha_2)g_{42}^2P}{g_{42}^2\alpha_2P+1}\right)\\
%&+3R_0-C_1-2C_2\nonumber
\end{align*}
and $R_1+2R_2$ is bounded by (32)-(37) with indices 1 and 2 switched where $\C(x)\triangleq\frac{1}{2}\log(1+x)$, $a_1=g_{31}g_{42}-g_{32}g_{41}$, $a_2=g_{31}g_{52}-g_{32}g_{51}$, $b_{11}=g_{31}^2+(1+\sigma^2)g_{41}^2$, $b_{12}=g_{32}^2+(1+\sigma^2)g_{42}^2$, $b_{21}=g_{31}^2+(1+\sigma^2)g_{51}^2$, $b_{22}=g_{32}^2+(1+\sigma^2)g_{52}^2$, $C_1=\C\left(\frac{(g_{32}^2+g_{42}^2)\alpha_2P+1}{(g_{42}^2\alpha_2P+1)\sigma^2}\right)$, and $C_2=\C\left(\frac{(g_{31}^2+g_{51}^2)\alpha_1P+1}{(g_{51}^2\alpha_1P+1)\sigma^2}\right)$ for some $\sigma^2>0$.\\

\begin{remark}
Above inner bound is the same as that achieved by Han-Kobayashi and Generalized Hash-and-Forward schemes \cite{Yu:arxiv11} for a GIRC with a digital relay link of rate $R_0$ bits per channel use.\\
\end{remark}

\begin{remark}
The sum rates of the proposed scheme is compared with that of \cite{Lim:10} in Fig. \ref{fig:GIRC_sum}. The sum-rate curve for the noisy network coding from \cite{Lim:10} is obtained by using noisy network coding via simultaneous nonunique decoding and that via treating interference as noise. Proposed scheme outperforms these two schemes since Han-Kobayashi scheme is more general and includes as special cases both simultaneous nonunique decoding and treating interference as noise.
\begin{figure}[t]
 \centering
  { \includegraphics[width=75mm]{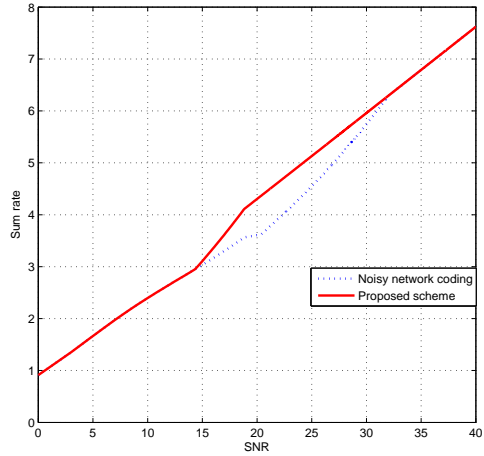}}
  \caption{Comparison of two schemes for the GIRC with orthogonal receiver components with $g_{41}=g_{52}=1, g_{42}=g_{51}=0.4, g_{31}=0.5, g_{32}=0.1, R_0=1, \sigma^2=5$.} \label{fig:GIRC_sum}
\end{figure}\\
\end{remark}

\appendix{
\subsection{Proof of Theorem \ref{thm:mib}}  \label{app:mib}

In this section, $\delta(\epsilon)>0$ denotes a function of $\epsilon>0$ which tends to zero as $\epsilon\rightarrow 0$ and $\epsilon_n> 0$ denotes a function of $n$ which tends to zero as $n\rightarrow \infty$. Moreover, we define $[i:j]\triangleq \{i,i+1,\ldots,j\}$ where $i$ and $j$ are integers. By using the coded time-sharing technique \cite{GamalKim:11}, achievability with $Q$ can be obtained. Hence it suffices to consider the case $Q=\emptyset$. Some of our proof steps and notations are based on those in Section 6.5 in \cite{GamalKim:11} and those in \cite{Lim:10}.

\subsubsection{Codebook generation} Fix a joint distribution of $p(u_1,x_1)p(u_2,x_2)\prod_{k=1}^N p(x_{3,k})p(\hat{y}_{3,k}|y_{3,k},x_{3,k})$. Let $b$ denote the number of blocks. For $a=1,2$, $m_a$ denotes the message for the $a$th source. For each message, we use rate splitting $m_a=(m_{a0},m_{aa})$ where $m_{a0}$ is the \emph{public} message at rate $R_{a0}$ and $m_{aa}$ is the \emph{private} message at rate $R_{aa}$ such that $R_a=R_{a0}+R_{aa}$. Moreover, for $k\in[1:N]$, $l_{kj}$ denotes the $k$th relay's compression index of its received signal in block $j$ at rate $\hat{R}_{3,k}$.

For each block $j\in [1:b]$ and each node $a=1,2$, randomly and independently generate $2^{nbR_{a0}}$ length-$n$ sequences $u_{aj}(m_{a0}), m_{a0}\in [1:2^{nbR_{a0}}],$ each according to the distribution $\prod_{i=1}^n p_{U_a}(u_{a,(j-1)n+i})$. For each $m_{a0}$, randomly and conditionally independently generate $2^{nbR_{aa}}$ sequences $x_{aj}(m_{a0},m_{aa}), m_{aa}\in [1:2^{nbR_{aa}}],$ each according to the distribution $\prod_{i=1}^n p_{X_a|U_a}(x_{a,(j-1)n+i}|u_{a,(j-1)n+i}(m_{a0}))$. Similarly, randomly and independently generate $2^{n\hat{R}_{3,k}}$ sequences $x_{3,kj}(l_{k,j-1}), l_{k,j-1}\in [1:2^{n\hat{R}_{3,k}}],$ each according to the distribution $\prod_{i=1}^n p_{X_{3,k}}(x_{3,k,(j-1)n+i})$. For each $x_{3,kj}(l_{k,j-1}), l_{k,j-1}\in [1:2^{n\hat{R}_{3,k}}],$ randomly and conditionally independently generate $2^{n\hat{R}_{3,k}}$ sequences $\hat{y}_{3,kj}(l_{kj}|l_{k,j-1}), l_{kj}\in [1:2^{n\hat{R}_{3,k}}]$, each according to $\prod_{i=1}^n p_{\hat{Y}_{3,k}|X_{3,k}}(\hat{y}_{3,k,(j-1)n+i}|x_{3,k,(j-1)n+i}(l_{k,j-1}))$. The codebook is defined as
\begin{align*}
\mathcal{C}_j=\{&u_{1j}(m_{10}),u_{2j}(m_{20}),x_{1j}(m_{10},m_{11}),x_{2j}(m_{20},m_{22}),x_{3,kj}(l_{k,j-1}),\hat{y}_{3,kj}(l_{kj}|l_{k,j-1}):\\
&m_{10}\in [1:2^{nbR_{10}}],m_{11}\in [1:2^{nbR_{11}}],m_{20}\in [1:2^{nbR_{20}}],m_{22}\in [1:2^{nbR_{22}}],\\
&l_{kj},l_{k,j-1}\in [1:2^{n\hat{R}_{3,k}}],k\in[1:N]\}
\end{align*}
for $j\in [1:b]$.

\subsubsection{Encoding} To send $m_a=(m_{a0},m_{aa})$, source $a$ transmits $x_{aj}^n(m_{a0},m_{aa})$ in block $j$. Set $l_{k0}=1, k\in[1:N]$ by convention. Upon receiving $y_{3,kj}^n$ at the end of block $j\in [1:b]$, the $k$th relay finds an index $l_{kj}$ such that
\begin{equation*}
(\hat{y}_{3,kj}^n(l_{kj}|l_{k,j-1}),y_{3,kj}^n,x_{3,kj}^n(l_{k,j-1}))\in \mathcal{T}_{\epsilon^{\prime}}^{(n)},
\end{equation*}
If there is more than one such index, select one of them uniformly at random. On the other hand, if there is no such index, choose an arbitrary index from $[1:2^{n\hat{R}_{3,k}}]$ uniformly at random. Then the $k$th relay transmits the codeword $x_{3,kj}^n(l_{k,j-1})$ in block $j\in [1:b]$.\\

\subsubsection{Decoding} Let $\epsilon>\epsilon^{\prime}$. We use simultaneous nonunique decoding. At the end of block $b$, node 4 finds the unique message pair $(\hat{m}_{10},\hat{m}_{11})$ such that there exist some $m_{20}\in[1:2^{nbR_{20}}]$ and $(\hat{l}_{1j},\cdots,\hat{l}_{Nj})$ satisfying
\begin{align*}
(&u_{1j}(\hat{m}_{10}),u_{2j}(m_{20}),x_{1j}(\hat{m}_{10},\hat{m}_{11}),x_{3,1j}(\hat{l}_{1,j-1}),\cdots,x_{3,Nj}(\hat{l}_{1,N,j-1}),\hat{y}_{3,1j}(\hat{l}_{1j}|\hat{l}_{1,j-1}),\cdots,\\
&\hat{y}_{3,Nj}(\hat{l}_{Nj}|\hat{l}_{N,j-1}),y_{4j})\in \mathcal{T}_{\epsilon}^{(n)}
\end{align*}
for all $j\in [1:b]$, where $\hat{m}_{10}\in [1:2^{nbR_{10}}], \hat{m}_{11}\in [1:2^{nbR_{11}}]$ and $\hat{l}_{kj}\in [1:2^{n\hat{R}_{3,k}}]$. Otherwise, it declares an error. Similarly, node 5 finds the message pair $(\hat{m}_{20},\hat{m}_{22})$.\\

\subsubsection{Analysis of the probability of error} Without loss of generality, assume that the message pair $((1,1),(1,1))$ and index $\mathbf{l}^{b}=(\mathbf{l}_{1},\ldots,\mathbf{l}_{b})=(\mathbf{1},\ldots,\mathbf{1})$ are sent where $\mathbf{l}_j=(l_{1j},\cdots,l_{Nj})$. Then the decoders make an error only if one of the following events occur:
\begin{align*}
\mathcal{E}_1=\{(&\hat{Y}_{3,kj}(l_{kj}|1),X_{3,kj}(1),Y_{3,kj})\notin \mathcal{T}_{\epsilon^{\prime}}^{(n)}\mbox{ for all }l_{kj}\mbox{ and for some }j\in [1:b],k\in[1:N]\}\\
\mathcal{E}_2=\{(&U_{1j}(1),U_{2j}(1),X_{1j}(1,1),X_{3,1j}(1),\cdots,X_{3,Nj}(1),\hat{Y}_{3,1j}(1|1),\cdots,\hat{Y}_{3,Nj}(1|1),Y_{4j})\notin \mathcal{T}_{\epsilon}^{(n)}\\
&\mbox{ for some }j\in [1:b]\}\\
\mathcal{E}_3=\{(&U_{1j}(1),U_{2j}(1),X_{2j}(1,1),X_{3,1j}(1),\cdots,X_{3,Nj}(1),\hat{Y}_{3,1j}(1|1),\cdots,\hat{Y}_{3,Nj}(1|1),Y_{5j})\notin \mathcal{T}_{\epsilon}^{(n)}\\
&\mbox{ for some }j\in [1:b]\}\\
\mathcal{E}_4=\{(&U_{1j}(1),U_{2j}(1),X_{1j}(1,m_{11}),X_{3,1j}(l_{1,j-1}),\cdots,X_{3,Nj}(l_{N,j-1}),\hat{Y}_{3,1j}(l_{1j}|l_{1,j-1}),\cdots,\\
&\hat{Y}_{3,Nj}(l_{Nj}|l_{N,j-1}),Y_{4j})\in \mathcal{T}_{\epsilon}^{(n)}\mbox{ for all }j\in [1:b]\mbox{ and for some }\mathbf{l}^b,m_{11}\neq 1\}\\
\mathcal{E}_5=\{(&U_{1j}(m_{10}),U_{2j}(1),X_{1j}(m_{10},m_{11}),X_{3,1j}(l_{1,j-1}),\cdots,X_{3,Nj}(l_{N,j-1}),\hat{Y}_{3,1j}(l_{1j}|l_{1,j-1}),\cdots,\\
&\hat{Y}_{3,Nj}(l_{Nj}|l_{N,j-1}),Y_{4j})\in \mathcal{T}_{\epsilon}^{(n)}\mbox{ for all }j\in [1:b]\mbox{ and for some }\mathbf{l}^b,m_{10}\neq 1,m_{11}\}\\
\mathcal{E}_6=\{(&U_{1j}(1),U_{2j}(m_{20}),X_{1j}(1,m_{11}),X_{3,1j}(l_{1,j-1}),\cdots,X_{3,Nj}(l_{N,j-1}),\hat{Y}_{3,1j}(l_{1j}|l_{1,j-1}),\cdots,\\
&\hat{Y}_{3,Nj}(l_{Nj}|l_{N,j-1}),Y_{4j})\in \mathcal{T}_{\epsilon}^{(n)}\mbox{ for all }j\in [1:b]\mbox{ and for some }\mathbf{l}^b,m_{20}\neq 1,m_{11}\neq 1\}\\
\mathcal{E}_7=\{(&U_{1j}(m_{10}),U_{2j}(m_{20}),X_{1j}(m_{10},m_{11}),X_{3,1j}(l_{1,j-1}),\cdots,X_{3,Nj}(l_{N,j-1}),\hat{Y}_{3,1j}(l_{1j}|l_{1,j-1}),\cdots,\\
&\hat{Y}_{3,Nj}(l_{Nj}|l_{N,j-1}),Y_{4j})\in \mathcal{T}_{\epsilon}^{(n)}\mbox{ for all }j\in [1:b]\mbox{ and for some }\mathbf{l}^b,m_{10}\neq 1,m_{20}\neq 1,m_{11}\}\\
\mathcal{E}_8=\{(&U_{1j}(1),U_{2j}(1),X_{2j}(1,m_{22}),X_{3,1j}(l_{1,j-1}),\cdots,X_{3,Nj}(l_{N,j-1}),\hat{Y}_{3,1j}(l_{1j}|l_{1,j-1}),\cdots,\\
&\hat{Y}_{3,Nj}(l_{Nj}|l_{N,j-1}),Y_{5j})\in \mathcal{T}_{\epsilon}^{(n)}\mbox{ for all }j\in [1:b]\mbox{ and for some }\mathbf{l}^b,m_{22}\neq 1\}\\
\mathcal{E}_9=\{(&U_{1j}(1),U_{2j}(m_{20}),X_{2j}(m_{20},m_{22}),X_{3,1j}(l_{1,j-1}),\cdots,X_{3,Nj}(l_{N,j-1}),\hat{Y}_{3,1j}(l_{1j}|l_{1,j-1}),\cdots,\\
&\hat{Y}_{3,Nj}(l_{Nj}|l_{N,j-1}),Y_{5j})\in \mathcal{T}_{\epsilon}^{(n)}\mbox{ for all }j\in [1:b]\mbox{ and for some }\mathbf{l}^b,m_{20}\neq 1,m_{22}\}\\
\mathcal{E}_{10}=\{(&U_{1j}(m_{10}),U_{2j}(1),X_{2j}(1,m_{22}),X_{3,1j}(l_{1,j-1}),\cdots,X_{3,Nj}(l_{N,j-1}),\hat{Y}_{3,1j}(l_{1j}|l_{1,j-1}),\cdots,\\
&\hat{Y}_{3,Nj}(l_{Nj}|l_{N,j-1}),Y_{5j})\in \mathcal{T}_{\epsilon}^{(n)}\mbox{ for all }j\in [1:b]\mbox{ and for some }\mathbf{l}^b,m_{10}\neq 1,m_{22}\neq 1\}\\
\mathcal{E}_{11}=\{(&U_{1j}(m_{10}),U_{2j}(m_{20}),X_{2j}(m_{20},m_{22}),X_{3,1j}(l_{1,j-1}),\cdots,X_{3,Nj}(l_{N,j-1}),\hat{Y}_{3,1j}(l_{1j}|l_{1,j-1}),\cdots,\\
&\hat{Y}_{3,Nj}(l_{Nj}|l_{N,j-1}),Y_{5j})\in \mathcal{T}_{\epsilon}^{(n)}\mbox{ for all }j\in [1:b]\mbox{ and for some }\mathbf{l}^b,m_{10}\neq 1,m_{20}\neq 1,m_{22}\}
\end{align*}

Thus, the probability of error is bounded as
\begin{align*}
\textsf{P}(\mathcal{E})\leq &\textsf{P}(\mathcal{E}_1)+\textsf{P}(\mathcal{E}_2\cap\mathcal{E}_1^c)+\textsf{P}(\mathcal{E}_3\cap\mathcal{E}_1^c)+\textsf{P}(\mathcal{E}_4)+\textsf{P}(\mathcal{E}_5)\\
&+\textsf{P}(\mathcal{E}_6)+\textsf{P}(\mathcal{E}_7)+\textsf{P}(\mathcal{E}_8)+\textsf{P}(\mathcal{E}_9)+\textsf{P}(\mathcal{E}_{10})+\textsf{P}(\mathcal{E}_{11}).
\end{align*}

By the covering lemma \cite{GamalKim:11} and the union of events bound over $b$ blocks, $\textsf{P}(\mathcal{E}_1)$ tends to zero as $n\rightarrow\infty$ if $\hat{R}_{3,k}>I(\hat{Y}_{3,k};Y_{3,k}|X_{3,k})+\delta(\epsilon^{\prime}), k\in[1:N]$. Next, using the Markov lemma \cite{GamalKim:11} and the union of events bound over $b$ blocks, the second and third terms $\textsf{P}(\mathcal{E}_2\cap\mathcal{E}_1^c), \textsf{P}(\mathcal{E}_3\cap\mathcal{E}_1^c)$ tend to zero as $n\rightarrow\infty$.

From here, we use similar proof techniques and steps as in \cite{Lim:10}. For the fourth term, by defining the events
\begin{align*}
\tilde{\mathcal{E}}_j(m_{11},\mathbf{l}_{j-1},\mathbf{l}_j)=\{(&U_{1j}(1),U_{2j}(1),X_{1j}(1,m_{11}),X_{3,1j}(l_{1,j-1}),\cdots,X_{3,Nj}(l_{N,j-1}),\\
&\hat{Y}_{3,1j}(l_{1j}|l_{1,j-1}),\cdots,\hat{Y}_{3,Nj}(l_{Nj}|l_{N,j-1}),Y_{4j})\in \mathcal{T}_{\epsilon}^{(n)}\},
\end{align*}

we can show
\begin{align*}
\textsf{P}(\mathcal{E}_4)&=\textsf{P}(\bigcup_{m_{11}\neq 1}\bigcup_{\mathbf{l}^b}\bigcap_{j=1}^b\tilde{\mathcal{E}}_j(m_{11},\mathbf{l}_{j-1},\mathbf{l}_j))\\
&\leq \sum_{m_{11}\neq 1}\sum_{\mathbf{l}^b}\textsf{P}(\bigcap_{j=1}^b\tilde{\mathcal{E}}_j(m_{11},\mathbf{l}_{j-1},\mathbf{l}_j))\\
&\overset{(a)}{=} \sum_{m_{11}\neq 1}\sum_{\mathbf{l}^b}\prod_{j=1}^b\textsf{P}(\tilde{\mathcal{E}}_j(m_{11},\mathbf{l}_{j-1},\mathbf{l}_j))\\
&\leq \sum_{m_{11}\neq 1}\sum_{\mathbf{l}^b}\prod_{j=2}^b\textsf{P}(\tilde{\mathcal{E}}_j(m_{11},\mathbf{l}_{j-1},\mathbf{l}_j)),
\end{align*}
where $(a)$ follows since the channel is memoryless and the codebook is independently generated for each block $j$.

For each $\mathbf{l}^b$ and $j\in[2:b]$, define $S_j(\mathbf{l}^b)=\{k\in[1:N]:l_{k,j-1}\neq 1\}$. Note that $S_j(\mathbf{l}^b)$ depends only on $\mathbf{l}_{j-1}$ and hence we write it as $S_j(\mathbf{l}_{j-1})$. Define $X_{3,j}(S_j(\mathbf{l}_{j-1}))$ to be the set of $X_{3,kj}(l_{k,j-1}),k\in S_j(\mathbf{l}_{j-1}),$ where $l_{k,j-1}$ is the corresponding element in $\mathbf{l}^b$. Similarly define $\hat{Y}_{3,j}(S_j(\mathbf{l}_{j-1}))$ and $Y_{3,j}(S_j(\mathbf{l}_{j-1}))$. Then, for $m_{11}\neq 1$,
\begin{align*}
&(U_{1j}(1),U_{2j}(1),X_{1j}(1,m_{11}),X_{3,j}(S_j(\mathbf{l}_{j-1})),\hat{Y}_{3,j}(S_j(\mathbf{l}_{j-1})))\\
&\sim\prod_{i=1}^n P_{U_1,X_1}(u_{1,(j-1)n+i},x_{1,(j-1)n+i}) P_{U_2}(u_{2,(j-1)n+i})\times\\
&~~~~~~~~\prod_{k\in S_j}P_{X_{3,k}}(x_{3,k,(j-1)n+i})P_{\hat{Y}_{3,k}|X_{3,k}}(\hat{y}_{3,k,(j-1)n+i}|x_{3,k,(j-1)n+i})
\end{align*}
is independent of $(X_{3,j}(S_j^c(\mathbf{l}_{j-1})),\hat{Y}_{3,j}(S_j^c(\mathbf{l}_{j-1})),Y_{4j})$. By the joint typicality lemma \cite{GamalKim:11} or Lemma 2 in \cite{Lim:10}, we have
\begin{equation*}
\textsf{P}(\tilde{\mathcal{E}}_j(m_{11},\mathbf{l}_{j-1},\mathbf{l}_j))\leq 2^{-n(A(S_j(\mathbf{l}_{j-1}))-\delta(\epsilon))},
\end{equation*}
where \begin{equation*}
A(S)=I(Y_4,\hat{Y}_3(S^c);X_1,X_3(S)|U_1,U_2,X_3(S^c))+\sum_{k\in S}I(\hat{Y}_{3,k};X_1,U_2,X_3^N,Y_4,\hat{Y}_3(S^c),\hat{Y}_3^{k-1}|X_{3,k}).
\end{equation*}

Furthermore, by the definition of $S_j(\mathbf{l}_{j-1})$, if $m_{11}\neq 1$, then
\begin{align*}
\sum_{\mathbf{l}_{j-1}} 2^{-n(A(S_j(\mathbf{l}_{j-1}))-\delta(\epsilon))}&=\sum_{S\subset[1:N]}\sum_{\mathbf{l}_{j-1}:S_j(\mathbf{l}_{j-1})=S} 2^{-n(A(S_j(\mathbf{l}_{j-1}))-\delta(\epsilon))}\\
&\leq\sum_{S\subset[1:N]} 2^{-n(A(S)-\sum_{k\in S} \hat{R}_{3,k}-\delta(\epsilon))}\\
&\leq 2^N 2^{-n(\min_S(A(S)-\sum_{k\in S} \hat{R}_{3,k}-\delta(\epsilon)))},
\end{align*}
where the minimum is over $S\subset[1:N]$. Hence,
\begin{align}
\textsf{P}(\mathcal{E}_4)&\leq\sum_{m_{11}\neq 1}\sum_{\mathbf{l}^b}\prod_{j=2}^b\textsf{P}(\tilde{\mathcal{E}}_j(m_{11},\mathbf{l}_{j-1},\mathbf{l}_j))\nonumber\\
&=\sum_{m_{11}\neq 1}\sum_{\mathbf{l}_b}\sum_{\mathbf{l}^{b-1}}\prod_{j=2}^b\textsf{P}(\tilde{\mathcal{E}}_j(m_{11},\mathbf{l}_{j-1},\mathbf{l}_j))\nonumber\\
&\leq\sum_{m_{11}\neq 1}\sum_{\mathbf{l}_b}\sum_{\mathbf{l}^{b-1}}\prod_{j=2}^b 2^{-n(A(S_j(\mathbf{l}_{j-1}))-\delta(\epsilon))}\nonumber\\
&=\sum_{m_{11}\neq 1}\sum_{\mathbf{l}_b}\prod_{j=2}^b\left(\sum_{\mathbf{l}_{j-1}} 2^{-n(A(S_j(\mathbf{l}_{j-1}))-\delta(\epsilon))}\right)\nonumber\\
&\leq 2^{nbR_{11}}2^{\sum_{k=1}^N n\hat{R}_{3,k}}2^{N(b-1)}\cdot 2^{n\left(-(b-1)\min_S(A(S)-\sum_{k\in S} \hat{R}_{3,k}-\delta(\epsilon))\right)},
\end{align}
where the minimum in $(1)$ is over all $S\subset[1:N]$. Then, $(1)$ tends to zero as $n\rightarrow\infty$ if
\begin{align*}
R_{11}<&\frac{b-1}{b}\left(\min_S\left(A(S)-\sum_{k\in S} \hat{R}_{3,k}\right)-\delta^{\prime}(\epsilon)\right)-\frac{1}{b}\sum_{k=1}^N \hat{R}_{3,k}
\end{align*}
for all $S\subset[1:N]$. By eliminating $\hat{R}_{3,k}>I(\hat{Y}_{3,k};Y_{3,k}|X_{3,k})+\delta(\epsilon^{\prime})$ and taking $b\rightarrow\infty$, it can be readily shown that $\textsf{P}(\mathcal{E}_4)$ tends to zero as $n\rightarrow\infty$ if
\begin{align*}
R_{11}<&\min_{S} A(S)-\sum_{k\in S} I(\hat{Y}_{3,k};Y_{3,k}|X_{3,k})-\delta^{\prime}(\epsilon)-N\delta(\epsilon^{\prime})\\
=&\min_{S}\{I(Y_4,\hat{Y}_3(S^c);X_1,X_3(S)|U_1,U_2,X_3(S^c))\\
&-\sum_{k\in S}I(\hat{Y}_{3,k};Y_{3,k}|X_1,U_2,X_3^N,\hat{Y}_3(S^c),Y_4,\hat{Y}_3^{k-1})\}-\delta^{\prime}(\epsilon)-N\delta(\epsilon^{\prime})\\
=&\min_{S}\{I(Y_4,\hat{Y}_3(S^c);X_1,X_3(S)|U_1,U_2,X_3(S^c))\\
&-\sum_{k\in S}I(\hat{Y}_{3,k};Y_{3}(S)|X_1,U_2,X_3^N,\hat{Y}_3(S^c),Y_4,\hat{Y}_3^{k-1})\}-\delta^{\prime}(\epsilon)-N\delta(\epsilon^{\prime})\\
=&\min_{S}\{I(Y_4,\hat{Y}_3(S^c);X_1,X_3(S)|U_1,U_2,X_3(S^c))\\
&-I(\hat{Y}_{3}(S);Y_{3}(S)|X_1,U_2,X_3^N,\hat{Y}_3(S^c),Y_4)\}-\delta^{\prime}(\epsilon)-N\delta(\epsilon^{\prime})
\end{align*}
for all $S\subset[1:N]$.

Similarly, $\textsf{P}(\mathcal{E}_5),\textsf{P}(\mathcal{E}_6)$, and $\textsf{P}(\mathcal{E}_7)$ tend to zero as $n\rightarrow \infty$ if the following conditions are satisfied respectively:\\
\begin{align*}
R_{11}+R_{10}<\min_S &I(X_1,X_3(S);\hat{Y}_3(S^c),Y_4|U_2,X_3(S^c))\\
&-I(\hat{Y}_3(S);Y_3(S)|X_1,U_2,X_3^N,\hat{Y}_3(S^c),Y_4)-\delta^{\prime}(\epsilon)-N\delta(\epsilon^{\prime})\\
R_{11}+R_{20}<\min_S &I(X_1,U_2,X_3(S);\hat{Y}_3(S^c),Y_4|U_1,X_3(S^c))\\
&-I(\hat{Y}_3(S);Y_3(S)|X_1,U_2,X_3^N,\hat{Y}_3(S^c),Y_4)-\delta^{\prime}(\epsilon)-N\delta(\epsilon^{\prime})\\
R_{11}+R_{10}+R_{20}<\min_S &I(X_1,U_2,X_3(S);\hat{Y}_3(S^c),Y_4|X_3(S^c))\\
&-I(\hat{Y}_3(S);Y_3(S)|X_1,U_2,X_3^N,\hat{Y}_3(S^c),Y_4)-\delta^{\prime}(\epsilon)-N\delta(\epsilon^{\prime}).
\end{align*}
Similarly, $\textsf{P}(\mathcal{E}_8),\textsf{P}(\mathcal{E}_9)$, $\textsf{P}(\mathcal{E}_{10})$, and $\textsf{P}(\mathcal{E}_{11})$ are bounded. Finally, we obtain all inequalities in Theorem \ref{thm:mib} by substituting $R_{11}=R_1-R_{10}$ and $R_{22}=R_2-R_{20}$, and applying the Fourier-Motzkin elimination.\\ \endproof

\subsection{Converse proof of Theorem \ref{thm:capacity}}\label{app:thm2}
Let $Q$ be a random variable uniformly distributed over $[1:n]$ and
independent of $(X_1^n,X_2^n,X_3^n,Y_3^n,Y_4^n,Y_5^n)$ and let
$X_1\triangleq X_{1Q}, X_2\triangleq X_{2Q}, X_3\triangleq X_{3Q},
Y_3\triangleq Y_{3Q}, Y_4\triangleq Y_{4Q}, Y_5\triangleq Y_{5Q},
T_1\triangleq T_{1Q}$, and $T_2\triangleq T_{2Q}$. The first term in the minimum in (2) in Theorem \ref{thm:capacity} is obtained as follows.
\begin{align*}
nR_1\leq &I(M_1;Y_4^n)+n\epsilon_n\\
\leq &I(M_1;Y_4^n,Y_3^n)+n\epsilon_n\\
\leq &I(X_1^n;Y_4^n,Y_3^n)+n\epsilon_n\\
\overset{(a)}{\leq} &I(X_1^n;Y_4^n,Y_3^n|T_2^n)+n\epsilon_n\\
= &H(Y_4^n,Y_3^n|T_2^n)+n\epsilon_n\\
=&\sum_{i=1}^n H(Y_{4i},Y_{3i}|Y_3^{i-1},Y_4^{i-1},T_2^n)+n\epsilon_n\\
\overset{(b)}{=}&\sum_{i=1}^n H(Y_{4i},Y_{3i}|Y_3^{i-1},Y_4^{i-1},X_{3i},T_2^n)+n\epsilon_n\\
\leq&\sum_{i=1}^n H(Y_{4i},Y_{3i}|X_{3i},T_{2i})+n\epsilon_n\\
%= &\sum_{i=1}^n H(Y_{4},Y_{3}|X_{3},T_{2},Q)+n\epsilon_n\\
= &nH(Y_{4}^{\prime},Y_4^{\prime\prime},Y_{3}|X_{3},T_{2},Q)+n\epsilon_n\\
%= &nH(Y_{4}^{\prime},Y_{3}|X_{3},T_{2},Q)+nH(Y_{4}^{\prime\prime}|X_{3},T_{2},Y_3,Y_4^{\prime},Q)+n\epsilon_n\\
\leq &nH(Y_{4}^{\prime},Y_{3}|T_{2},Q)+n\epsilon_n\\
= &n(H(Y_4^{\prime}|T_2,Q)+H(Y_3|Y_4^{\prime},T_2,Q)+\epsilon_n)
\end{align*}
where $Q$ is the usual time-sharing random variable and $(a)$
follows by the fact that $T_2^n$ and $X_1^n$ are independent. $(b)$
follows since $X_{3i}$ is a function of $Y_3^{i-1}$. The second term in the minimum in (2) is obtained as
follows.
\begin{align*}
nR_1&\leq I(M_1;Y_4^n)+n\epsilon_n\\
&\leq I(X_1^n,X_3^n;Y_4^n|T_2^n)+n\epsilon_n\\
&= I(X_1^n;Y_4^n|T_2^n,X_3^n)+I(X_3^n;Y_4^n|T_2^n)+n\epsilon_n\\
%&\overset{(a)}{=}H(Y_4^{\prime n}|T_2^n, X_3^n)+I(X_3^n;Y_4^{\prime\prime n}|T_2^n)+n\epsilon_n\\
&=H(Y_4^{\prime n}|T_2^n, X_3^n)+I(X_3^n;Y_4^{\prime\prime n}|T_2^n)+n\epsilon_n\\
&\leq n(H(Y_4^{\prime}|T_2,Q)+R_0+\epsilon_n)
\end{align*}
%where $(a)$ follows by the fact that $H(Y_4^{\prime\prime n}|T_2^n,X_3^n, X_1^n)=0$ and $I(X_3^n;Y_4^{\prime
%n}|T_2^n, Y_4^{\prime\prime n})=0$.
Similarly, we can obtain inequality (3). Before we prove the remaining terms, we define mutual information terms $I_1,\ldots,I_{12}$ and show some inequalities for those terms.
\begin{align*}
I_1=I(M_1;Y_4^n,Y_3^n,T_1^n|T_2^n\!)&\leq H(T_1^n)+H(Y_4^n,Y_3^n|T_1^n,T_2^n)\\%-H(Y_4^n,Y_3^n|T_1^n,T_2^n,X_1^n)
%&\overset{(a)}{\leq} H(T_1^n)+\sum_{i=1}^n \{H(Y_{3i},Y_{4i}|Y_3^{i-1},Y_4^{i-1},X_{3i},T_1^n,T_2^n)\\
%&~~~~-H(Y_{4i}^{\prime\prime}|Y_{4}^{\prime\prime i-1},T_1^n,T_2^n,X_1^n,Y_4^{\prime n},Y_{3}^n,X_{3i})\}\\
&\leq H(T_1^n)+\sum_{i=1}^n \{H(Y_{3i},Y_{4i}|Y_3^{i-1},Y_4^{i-1},X_{3i},T_1^n,T_2^n)\\
&\leq H(T_1^n)+\sum_{i=1}^n H(Y_{3i},Y_{4i}^{\prime}|T_{1i},T_{2i})\\
&=H(T_1^n)+nH(Y_3,Y_4^{\prime}|T_1,T_2,Q).
\end{align*}
Next, %where $(a)$ follows by the channel conditions and the fact that $X_{3i}$ is a function of $Y_3^{i-1}$.
\begin{align*}
I_2=I(M_2;Y_3^n,Y_5^n\!)&\leq H(Y_3^n,Y_5^n)-H(Y_3^n,Y_5^n|X_2^n)\\
&\overset{(a)}{=}H(Y_3^{n},Y_5^{\prime n})-H(T_1^n)\\
&\leq nH(Y_3,Y_5^{\prime}|Q)-H(T_1^n)
\end{align*}
where $(a)$ follows by the channel conditions and the fact that $H(Y_{5}^{\prime n}|X_2^n)=H(T_1^n)$. Next,
\begin{align*}
I_3=I(M_1;Y_4^n,T_1^n|T_2^n)&\leq H(T_1^n)+I(X_1^n,X_3^n;Y_4^n|T_1^n,T_2^n)\\
&= H(T_1^n)+I(X_3^n;Y_4^{\prime n},Y_4^{\prime\prime n}|T_1^n,T_2^n)+I(X_1^n;Y_4^n|X_3^n,T_1^n,T_2^n)\\
%&\overset{(a)}{\leq} H(T_1^n)+nR_0+nH(Y_4^{\prime}|T_1,T_2,Q)
&\leq H(T_1^n)+nR_0+nH(Y_4^{\prime}|T_1,T_2,Q)
\end{align*}
%where $(a)$ follows by the fact that $R_0=H(Y_4^{\prime\prime})=H(Y_5^{\prime\prime})$, $I(X_3^n;Y_4^{\prime n}|Y_4^{\prime\prime n},T_1^n,T_2^n)=0$ and\\
%$H(Y_4^{n}|X_3^n,T_1^n,T_2^n,X_1^n)=0$.
Next,
\begin{align*}
I_4=I(M_2;Y_5^n)&\leq I(X_2^n;Y_5^{\prime n})+I(X_2^n;Y_5^{\prime\prime n}|Y_5^{\prime n})\\
&\leq I(X_2^n;Y_5^{\prime n})+I(X_2^n,X_3^n;Y_5^{\prime\prime n}|Y_5^{\prime n})\\
&=H(Y_5^{\prime n})-H(T_1^n)+H(Y_5^{\prime\prime n}|Y_5^{\prime n})-H(Y_5^{\prime\prime n}|Y_5^{\prime n},X_2^n,X_3^n)\\
&\leq nH(Y_5^{\prime}|Q)-H(T_1^n)+nR_0.
\end{align*}
Similarly, $I_5,I_6,I_7,I_8$ are defined as $I_1,I_2,I_3,I_4$ with indices 1 and 2 switched. Next,
\begin{align*}
I_9=I(M_1;Y_4^n,Y_3^n,T_1^n)&\leq I(X_1^n;Y_4^n,Y_3^n,T_1^n)\\
&\overset{(a)}{=} H(T_1^n)+H(Y_3^n,Y_4^n|T_1^n)-H(T_2^n)\\
&= H(T_1^n)+\sum_{i=1}^n H(Y_{3i},Y_{4i}|Y_3^{i-1},Y_4^{i-1},T_1^n,X_{3i})-H(T_2^n)\\
&\leq H(T_1^n)+nH(Y_3,Y_4^{\prime}|T_1,Q)-H(T_2^n)
\end{align*}
where $(a)$ follows by the fact that $H(Y_{4}^{\prime n}|X_1^n)=H(T_2^n)$. Similarly,
\begin{align*}
I_{10}=I(M_2;Y_5^n,Y_3^n,T_2^n)\leq H(T_2^n)+nH(Y_3,Y_5^{\prime}|T_2,Q)-H(T_1^n).
\end{align*}
Next,
\begin{align*}
I_{11}=I(M_1;Y_4^n,T_1^n)&\leq I(X_1^n;Y_4^n,T_1^n)\\
&\leq H(T_1^n)+I(X_1^n;Y_4^{\prime n}|T_1^n)+I(X_1^n,X_3^n;Y_4^{\prime\prime n}|T_1^n,Y_4^{\prime n})\\
&=H(T_1^n)+H(Y_4^{\prime n}|T_1^n)-H(T_2^n)+H(Y_4^{\prime\prime n}|T_1^n,Y_4^{\prime n})\\
&\leq H(T_1^n)+nH(Y_4^{\prime}|T_1,Q)-H(T_2^n)+nR_0.
\end{align*}
Similarly,
\begin{align*}
I_{12}=I(M_2;Y_5^n,T_2^n)\leq H(T_2^n)+nH(Y_5^{\prime}|T_2,Q)-H(T_1^n)+nR_0.
\end{align*}
Then, the remaining terms in Theorem \ref{thm:capacity} can be proved by using above inequalities for $I_1,\ldots,I_{12}$. Inequality (4) is obtained as follows.
\begin{align*}
n(R_1+R_2)&\leq \min\{I(M_1;Y_4^n,Y_3^n,T_1^n|T_2^n\!),I(M_1;Y_4^n,T_1^n|T_2^n)\}\\
&~+\min\{I(M_2;Y_3^n,Y_5^n\!),I(M_2;Y_5^n)\}+n\epsilon_n\\
&=\min\{I_1,I_3\}+\min\{I_2,I_4\}+n\epsilon_n\\
&\leq \min\{H(T_1^n)+nH(Y_3,Y_4^{\prime}|T_1,T_2,Q),H(T_1^n)+nR_0+nH(Y_4^{\prime}|T_1,T_2,Q)\}\\
&~+\min\{nH(Y_3,Y_5^{\prime}|Q)-H(T_1^n),nH(Y_5^{\prime}|Q)-H(T_1^n)+nR_0\}+n\epsilon_n\\
&=n(\min\{H(Y_3,Y_4^{\prime}|T_1,T_2,Q),H(Y_4^{\prime}|T_1,T_2,Q)+R_0\}\\
&~+\min\{H(Y_3,Y_5^{\prime}|Q),H(Y_5^{\prime}|Q)+R_0\}+\epsilon_n).
\end{align*}
Similarly, we can obtain inequality (5) by using $I_5,I_6,I_7,I_8$. Inequality (6) is obtained as
follows.
\begin{align*}
n(R_1+R_2)&\leq \min\{I(M_1;Y_4^n,Y_3^n,T_1^n),I(M_1;Y_4^n,T_1^n)\}\\
&~+\min\{I(M_2;Y_5^n,Y_3^n,T_2^n),I(M_2;Y_5^n,T_2^n)\}+n\epsilon_n\\
&=\min\{I_9,I_{11}\}+\min\{I_{10},I_{12}\}+n\epsilon_n\\
&\leq \min\{H(T_1^n)+nH(Y_3,Y_4^{\prime}|T_1,Q)-H(T_2^n),H(T_1^n)+nH(Y_4^{\prime}|T_1,Q)-H(T_2^n)+nR_0\}\\
&~+\min\{H(T_2^n)+nH(Y_3,Y_5^{\prime}|T_2,Q)-H(T_1^n),H(T_2^n)+nH(Y_5^{\prime}|T_2,Q)-H(T_1^n)+nR_0\}\\
&~+n\epsilon_n\\
&=n(\min\{H(Y_3,Y_4^{\prime}|T_1,Q),H(Y_4^{\prime}|T_1,Q)+R_0\}\\
&~+\min\{H(Y_3,Y_5^{\prime}|T_2,Q),H(Y_5^{\prime}|T_2,Q)+R_0\}+\epsilon_n).
\end{align*}
Inequality (7) is obtained as
follows.
\begin{align*}
n(2R_1+R_2)&\leq 2I(M_1;Y_4^n)+I(M_2;Y_5^n)+n\epsilon_n\\
&\leq \min\{I(M_1;Y_4^n,Y_3^n,T_1^n|T_2^n\!),I(M_1;Y_4^n,T_1^n|T_2^n)\}\\
&~+\min\{I(M_1;Y_3^n,Y_4^n\!),I(M_1;Y_4^n)\}\\
&~+\min\{I(M_2;Y_5^n,Y_3^n,T_2^n),I(M_2;Y_5^n,T_2^n)\}+n\epsilon_n\\
&=\min\{I_1,I_3\}+\min\{I_6,I_8\}+\min\{I_{10},I_{12}\}+n\epsilon_n\\
&\leq \min\{H(T_1^n)+nH(Y_3,Y_4^{\prime}|T_1,T_2,Q),H(T_1^n)+nR_0+nH(Y_4^{\prime}|T_1,T_2,Q)\}\\
&~+\min\{nH(Y_3,Y_4^{\prime}|Q)-H(T_2^n),nH(Y_4^{\prime}|Q)-H(T_2^n)+nR_0\}\\
&~+\min\{H(T_2^n)+nH(Y_3,Y_5^{\prime}|T_2,Q)-H(T_1^n),H(T_2^n)+nH(Y_5^{\prime}|T_2,Q)-H(T_1^n)+nR_0\}\\
&~+n\epsilon_n\\
&=n(\min\{H(Y_3,Y_4^{\prime}|T_1,T_2,Q),H(Y_4^{\prime}|T_1,T_2,Q)+R_0\}\\
&~+\min\{H(Y_3,Y_4^{\prime}|Q),H(Y_4^{\prime}|Q)+R_0\}\\
&~+\min\{H(Y_3,Y_5^{\prime}|T_2,Q),H(Y_5^{\prime}|T_2,Q)+R_0\}+\epsilon_n).
\end{align*}
Similarly, we can obtain inequality (8). This completes the converse proof
for Theorem \ref{thm:capacity}. \endproof
}

\bibliographystyle{IEEEtran}
\bibliography{IEEEabrv,References}

\end{document}